\documentclass[11pt]{article}

\usepackage{amsmath}
\usepackage{amssymb}
\usepackage{amsthm}
\usepackage{mathtools}
\usepackage[shortlabels]{enumitem}

\usepackage{bm}              
\usepackage{bbm}             
\usepackage{mathrsfs}

\usepackage{graphicx}
\usepackage{float}
\usepackage{tikz}
\usetikzlibrary{arrows.meta,calc,shapes,decorations.pathmorphing}

\usepackage{subcaption}

\usepackage{authblk}

\usepackage{geometry}
\usepackage{hyperref}
\hypersetup{
	colorlinks = true,
	linkcolor  = blue,
	citecolor  = blue,
	urlcolor   = blue
}

\usepackage{cleveref}

\newtheorem{theorem}{Theorem}[section]

\usepackage{xcolor}

\usepackage{todonotes}

\begin{document}
	
	\title{{\bf ESAR: Event-Based Synthetic Aperture Reconstruction}\thanks{This work is partially supported by the Office of Naval Research (ONR) under Award NO: N00014-24-1-2147. NSF grant DMS-2408877, and the Air Force Office of Scientific Research (AFOSR) under Award NO: FA9550-25-1-0231}}

    \author[1]{Harbir Antil}
    \author[1]{Daniel Blauvelt}
    \author[1]{David Sayre}

    \affil[1]{Center for Mathematics and Artificial Intelligence and Department of Mathematical Sciences, George Mason University,  Fairfax, Virginia 22030}
    
	\date{}
	\maketitle
	
	\begin{abstract}
		Event cameras report asynchronous polarity events when changes in log--radiance exceed a fixed contrast threshold, producing signed temporal contrast measurements rather than conventional image frames. We formulate monocular event-based imaging as a synthetic-aperture inverse problem for a static ground-domain log--radiance field $\theta \in \mathbb{R}^{N_g}$. Instead of reconstructing a latent pixel-time volume $v \in \mathbb{R}^{N_pN_t}$, we impose the geometric relation $v=P\theta$, where $P$ maps the fixed scene into motion-dependent latent views. Aggregating events over finite time intervals gives the linearized model
		\[
		AP\theta = b+\eta,
		\]
		where $A$ is a temporal differencing operator, $b$ contains signed binned event counts, and $\eta$ represents measurement and modeling errors. This decomposition exposes a synthetic-aperture structure: under near-nadir motion, successive projections are approximately shifted views of a common scene, while the composite operator $AP$ remains ill-conditioned because it combines spatial averaging with temporal differencing. We therefore use regularized inversion to recover $\theta$. Numerical experiments on simulated data and real near-nadir Falcon Neuro event data show that the proposed $\theta$-based formulation recovers coherent large-scale spatial structure, relative to dynamic latent-image and learned event-reconstruction baselines, while suppressing fine-scale texture.
	\end{abstract}
	
	\section{Introduction}
	
	Neuromorphic event cameras differ fundamentally from conventional frame-based
	imagers in that they do not measure absolute intensity at fixed sampling
	intervals. Instead, each pixel independently reports asynchronous ``events''
	whenever the local change in log--radiance exceeds a fixed contrast threshold
	\cite{gallego2020survey}. This sensing paradigm enables extremely high temporal
	resolution, high dynamic range, and low latency, but it also complicates
	downstream inference: the sensor outputs a sparse stream of signed contrast
	changes rather than a sequence of images.
	
	A substantial body of recent work has addressed the reconstruction of intensity
	images or videos from event data. Existing approaches include deep learning
	methods
	\cite{Scheerlinck20wacv,paredes2021back,wang2019event,yang2023learning},
	variational and optimization-based formulations
	\cite{pan2019blurry,Antil_2023,zhang2022formulating}, and hybrid methods that
	fuse events with conventional frame measurements. While these methods have
	demonstrated strong empirical performance, many rely on learned priors,
	additional sensing modalities, or motion models that are difficult to interpret
	and analyze theoretically.
	
	Most of these approaches reconstruct a latent image or video sequence
	\[
	v \in \mathbb{R}^{N_pN_t},
	\]
	so that each pixel-time pair is treated as an independent unknown. This is also
	the perspective adopted in \cite{Antil_2023}, where the reconstruction target is
	a time-indexed latent signal constrained by event measurements. The present work
	instead parameterizes the unknown as a static ground-domain log--radiance field
	\[
	\theta \in \mathbb{R}^{N_g},
	\]
	with the latent sequence induced geometrically through
	\[
	v=P\theta.
	\]
	This parameterization is the central modeling distinction of the paper. It
	reduces the degrees of freedom from a pixel-time volume to a single scene
	representation, enforces geometric consistency across views through the
	projection operator $P$, and naturally extends to multi-pass or multi-sensor
	settings in which several trajectories observe the same underlying scene. While
	one could impose static-scene constraints inside a $v$-based formulation, the
	inverse problem would still be posed over the latent pixel-time volume. Here the
	static-scene assumption is built directly into the unknown, so the geometry and
	identifiability structure are expressed at the level of $\theta$.
	
	In parallel, synthetic aperture radar (SAR) and related imaging modalities have
	shown that high-resolution scene reconstruction can be achieved by coherently
	integrating multiple views acquired along a platform trajectory
	\cite{curlander1991synthetic,oliver2004understanding}. In SAR, scene resolution
	is not limited by the instantaneous aperture size, but is determined by the
	diversity of observations accumulated over time. This perspective raises the
	question of whether a similar synthetic-aperture interpretation can be developed
	for monocular event-based sensing, particularly when platform motion induces
	repeated observations of a static scene.
	
	We develop an operator-theoretic framework that casts monocular event sensing as
	a synthetic-aperture inverse problem. The forward model is separated into a
	temporal differencing operator, which encodes the event-generation mechanism,
	and a spatial projection operator, which maps a static ground-referenced
	log--radiance field into latent views along the platform trajectory. This gives
	the linearized model
	\[
	AP\theta=b+\eta,
	\]
	where $\theta$ is the unknown scene log--radiance, $P$ encodes platform geometry
	and footprint integration, $A$ captures temporal differencing induced by event
	generation, $b$ contains signed temporally binned event measurements, and $\eta$
	collects measurement and modeling errors.
	
	The formulation reveals a close analogy with strip-map SAR imaging. Under
	near-nadir motion and mild regularity assumptions on the projection operator,
	the operators $\{P_k\}$ are approximately shifted versions of a common
	blur-and-sample operator. Consequently, successive rows of $AP$ encode shifted
	differences of spatial footprints, producing a synthetic-aperture effect driven
	by platform motion rather than antenna geometry. At the same time, $AP$ is
	ill-conditioned because it combines spatial averaging through $P$ with temporal
	differencing through $A$, so stable reconstruction requires regularization.
	
	The contributions of this paper are threefold:
	\begin{itemize}
		\item We formulate monocular event-based synthetic-aperture reconstruction
		through the static-scene parameterization $v=P\theta$, yielding the
		composite forward model $AP\theta=b+\eta$ for a ground-domain
		log--radiance field $\theta \in \mathbb{R}^{N_g}$.
		\item We analyze the shift structure and ill-conditioning of the composite
		operator $AP$ and derive a regularized reconstruction framework for stable
		inversion.
		\item We present synthetic-aperture reconstructions on both simulated data
		and real near-nadir Falcon Neuro event data, a neuromorphic payload aboard
		the International Space Station \cite{mcharg2022falcon}, and compare the
		proposed $\theta$-based reconstructions against dynamic latent-image and
		learned event-reconstruction baselines.
	\end{itemize}
	
	By framing event-based imaging through the lens of synthetic aperture inversion,
	this work connects neuromorphic sensing, inverse problems, and SAR-style imaging,
	while remaining compatible with both model-based and learned extensions.
	
	\section{Forward Model for Event Synthetic Aperture}
	\label{sec:forward_model}
	
	We now build the forward model in one pass. We begin from the temporal
	event-modeling framework inherited from prior latent-sequence formulations,
	then impose the static-scene geometry through the projection operator $P$,
	which leads to the inverse problem
	\[
	AP\theta = b + \eta
	\]
	that we use throughout the remainder of the paper.
	
	\subsection{Temporal Event Model}
	\label{sec:temp_event_model}
	
	Event cameras do not measure absolute brightness. Instead, each pixel reports
	asynchronous events when the change in log--radiance exceeds a fixed contrast
	threshold \cite{gallego2020survey}. This subsection introduces a discrete temporal model for event sensing and defines the noise term $\eta$ that collects the resulting modeling errors.
	
	\paragraph{Log--Radiance and Threshold Events.}
	
	Let $q_\rho(t) > 0$ denote the true physical radiance at pixel $\rho$ and time
	$t \in [0,T]$, and define the corresponding true log--radiance
	\[
	\ell_\rho(t) := \log q_\rho(t).
	\]
	The positivity assumption ensures that the log--radiance is well defined.
	
	An event is generated at pixel $\rho$ whenever the accumulated change in
	$\ell_\rho(t)$ since the previous event crosses a fixed contrast threshold
	$c>0$. Let $\{s_j\}$ denote the (strictly increasing) sequence of event times at pixel $\rho$, and let $p(s_j)\in\{-1,+1\}$ denote the polarity of the event occurring at time $s_j$, indicating the sign of the change
	\cite{gallego2020survey,gallego2018contrast}. Individual events do not encode the exact magnitude of the log--radiance increment; they provide only signed threshold crossings.
	
	\paragraph{Binned Events and Additive Error Model.}
	
	Fix a temporal grid $0 = t_1 < t_2 < \cdots < t_{N_t}$.
	For each pixel $\rho$ and interval $(t_k,t_{k+1}]$, define the signed binned event measurement
	\[
	b_{\rho,k}
	=
	\sum_{s_j \in (t_k,t_{k+1}]} p(s_j).
	\]
	
	The net change in true log--radiance over the interval $(t_k,t_{k+1}]$ can be written as a sum of incremental changes between successive events occurring in that interval. To capture threshold variability, missed events, and temporal discretization in a single additive term, let $\eta_{\rho,k}$ denote the aggregate error over the interval $(t_k,t_{k+1}]$. Then
	\[
	\ell_\rho(t_{k+1}) - \ell_\rho(t_k)
	=
	c\,(b_{\rho,k} + \eta_{\rho,k}).
	\]
	
	Thus, while events do not provide direct access to absolute log--radiance values, their aggregation over a temporal bin yields the corresponding log--radiance difference up to the scaling factor $c$ and the additive error $\eta_{\rho,k}$.
	
	\paragraph{Latent Variables and Scaling.}
	
	Following prior work, we introduce discrete latent variables $v_{\rho,k}$.
	To obtain a clean linear model, we normalize by the contrast threshold and define
	\[
	v_{\rho,k} := \frac{\ell_\rho(t_k)}{c}.
	\]
	
	With this convention, the temporal relation becomes
	\[
	v_{\rho,k+1} - v_{\rho,k}
	=
	b_{\rho,k} + \eta_{\rho,k}.
	\]
	
	This normalization absorbs the contrast threshold into the latent variables and ensures that the forward model does not explicitly depend on $c$.
	
	\paragraph{Relation to Prior Latent-Sequence Models.}
	
	Prior work formulates binned event measurements through a latent time-indexed
	image sequence $v$ satisfying a temporal differencing relation of the form
	\[
	v_{\rho,k+1} - v_{\rho,k}
	=
	b_{\rho,k} + \eta_{\rho,k},
	\]
	an approach developed in the bilevel inverse formulation of
	\cite{Antil_2023} and the dynamic reconstruction framework of
	\cite{antil2024drnd}. We adopt this temporal event model as the starting point
	for our formulation, but depart from prior work in how the unknown is
	parameterized.
	
	Stacking $v_{\rho,k}$ across all pixels and times yields
	\[
	v \in \mathbb{R}^{N_p N_t},
	\]
	while stacking $b_{\rho,k}$ and $\eta_{\rho,k}$ yields
	\[
	b \in \mathbb{R}^{N_p (N_t-1)},
	\qquad
	\eta \in \mathbb{R}^{N_p (N_t-1)}.
	\]
	
	Let $A_{\mathrm{time}} \in \mathbb{R}^{(N_t-1)\times N_t}$ denote the forward
	difference operator and define
	\[
	A = A_{\mathrm{time}} \otimes I_{N_p}.
	\]
	Then the temporal model can be written compactly as
	\[
	Av = b + \eta.
	\]
	
	This temporal relation supplies the differencing operator $A$ inherited from
	prior latent-sequence models. We now incorporate the static-scene geometry
	through the projection operator $P$.
	
	\subsection{Static Scene Projection Model}
	\label{sec:projection}
	
	To encode the static-scene assumption directly in the unknown, we parameterize
	the scene by a ground-domain log--radiance field
	\[
	\theta \in \mathbb{R}^{N_g}
	\]
	and define the projection operator
	\[
	P : \mathbb{R}^{N_g} \to \mathbb{R}^{N_p N_t}.
	\]
	Here 
	Throughout, $\theta$ is understood in the same contrast-threshold-normalized
units as $v$ (Section~\ref{sec:temp_event_model}), i.e.\ physical
log--radiance divided by $c$, so that $v=P\theta$ requires no additional
scaling factor.
$P$ is the geometric component of the forward model: it maps the static scene
	parameter $\theta$ into the sequence of projected views to which the temporal
	differencing operator $A$ is applied.
	
	Instead of reconstructing the latent sequence $v$ directly, we impose the
	static-scene relation
	\[
	v = P\theta,
	\]
	which converts the prior temporal model into the composite synthetic-aperture
	equation
	\[
	AP\theta = b + \eta,
	\]
	used throughout the remainder of the paper. Here $\eta$ is understood to
	collect not only temporal event-model errors, but also the projection-side
	linearization and discretization errors introduced in this subsection. We now
	describe the imaging geometry that defines the rows of $P$ and the footprint
	structure of its time-indexed blocks $\{P_k\}$. Because each pixel footprint
	sweeps different ground regions as the camera moves, the stacked operator $P$
	exhibits a synthetic-aperture structure; its formal relationship to SAR is
	made precise in Section~\ref{sec:sar}.
	
	\paragraph{Imaging Geometry and Pixel Footprints.}
	
	We represent the ground log--radiance field on a fixed discretized grid consisting of $N_g$ surface elements. The unknown log--radiance is modeled as a finite-dimensional vector
	\[
	\theta = (\theta_1,\dots,\theta_{N_g})^\top \in \mathbb{R}^{N_g},
	\]
	where $\theta_i$ denotes the log--radiance associated with the $i$th surface cell.
	Each cell is represented by a point $X_i \in \mathbb{R}^2$, corresponding to its horizontal location in a world (ground) reference frame, for example the horizontal coordinates of the cell's geometric center.
	We orient this frame so that
	\[
	e_3=(0,0,1)^\top
	\]
	points in the vertical (altitude) direction, while the first two coordinates describe horizontal ground position.
	
	\medskip
	
	The camera is modeled as a calibrated pinhole system \cite{hartley_zisserman_2003}
	with intrinsic matrix
	\[
	K
	=
	\begin{pmatrix}
		f_x & 0   & c_x \\
		0   & f_y & c_y \\
		0   & 0   & 1
	\end{pmatrix}.
	\]
	Here $f_x$ and $f_y$ denote the effective focal lengths in pixel units along the horizontal and vertical image axes, and $(c_x,c_y)$ denotes the principal point (the projection of the camera center onto the image plane).
	
	At discrete time $t_k$, the camera pose is described by a rotation
	$R_k \in SO(3)$ and a platform position $C_k \in \mathbb{R}^3$, mapping world
	coordinates into the camera frame according to
	\[
	x^{\mathrm{cam}} = R_k (x^{\mathrm{world}} - C_k).
	\]
	
	Let $\rho$ denote a pixel index with image-plane coordinates
	\[
	\mathbf{x}_\rho = (x_\rho,y_\rho)^\top.
	\]
	Under the pinhole model, this pixel defines a ray with direction
	\[
	d_\rho
	=
	K^{-1}
	\begin{pmatrix}
		x_\rho\\
		y_\rho\\
		1
	\end{pmatrix}.
	\]
	
	This ray, together with the finite spatial extent of the pixel, induces a
	footprint on the ground surface. The sensor integrates radiance over this
	footprint before the logarithm enters the event model. To remain
	consistent with the log--radiance formulation above, we adopt a local
	linearization: when the radiance varies moderately across a localized
	footprint, the log of the footprint-averaged radiance is approximated by the
	corresponding weighted average of local log--radiance values. Under this
	approximation, the projected quantity $(P_k\theta)_\rho$ can be modeled
	directly from the ground log--radiance samples $\theta_i$ associated with cells whose
	representative points $X_i$ lie within or near the footprint. Any residual
	error from this linearization, together with footprint discretization error, is
	absorbed into the aggregate modeling term $\eta$ in the composite forward model
	introduced above.
	
	\paragraph{Footprint kernel and weights.}
	
	To model footprint averaging, we introduce a fixed kernel
	\[
	\kappa : \mathbb{R}^2 \to [0,\infty),
	\]
	which is assumed to be Lipschitz continuous, bounded, and compactly supported.
	
	For each pixel $\rho$ and time $t_k$, let $X_{\rho,k} \in \mathbb{R}^2$
	denote the center of the corresponding ground footprint. We define the
	associated footprint neighborhood as the discrete support of the kernel:
	\[
	\mathcal{N}_{\rho,k}
	=
	\{\, i \in \{1,\dots,N_g\} : \kappa(X_i - X_{\rho,k}) \neq 0 \,\}.
	\]
	
	We assume that the normalization denominator is uniformly bounded below,
	i.e., there exists $m>0$ such that
	\[
	\sum_{j\in\mathcal N_{\rho,k}} \kappa(X_j - X_{\rho,k}) \ge m
	\quad \text{for all } \rho,k.
	\]
	
	We define the normalized footprint weights by
	\[
	w_{\rho,k,i}
	=
	\frac{\kappa(X_i - X_{\rho,k})}
	{\sum_{j\in\mathcal N_{\rho,k}} \kappa(X_j - X_{\rho,k})},
	\qquad i \in \mathcal N_{\rho,k}.
	\]
	
	\paragraph{Construction and Sparsity of the Projection Operator.}
	
	For each time $t_k$, the projection operator
	\[
	P_k \in \mathbb{R}^{N_p \times N_g}
	\]
	maps $\theta$ to the corresponding latent image. Its $\rho$-th row is supported on $\mathcal{N}_{\rho,k}$, and is given by
	\[
	(P_k\theta)_\rho
	= \sum_{i\in\mathcal{N}_{\rho,k}} w_{\rho,k,i}\,\theta_i.
	\]
	Stacking the operators $P_k$ across all times yields
	\[
	P=
	\begin{pmatrix}
		P_1 \\ P_2 \\ \vdots \\ P_{N_t}
	\end{pmatrix}
	\in\mathbb{R}^{N_pN_t \times N_g}.
	\]
	Each row of each $P_k$ contains at most $|\mathcal{N}_{\rho,k}|$ nonzero entries, and therefore $P$ has at most
	\[
	N_pN_t \max_{\rho,k} |\mathcal{N}_{\rho,k}|
	\]
	nonzero entries. This sparsity is essential for efficient evaluation of the composite operator $AP$, and mirrors the sparsity patterns exploited in synthetic-aperture forward models
	\cite{curlander1991synthetic, jakowatz_sar_imaging}.
	
	\section{Synthetic-Aperture Structure of the Forward Operator}
	\label{sec:forward}
	
	With the composite forward model \(AP\theta = b + \eta\) established in Section~\ref{sec:forward_model}, we now analyze the structure of the signal operator \(AP\).
	This section develops the structural properties of \(AP\), establishes the approximate shift-invariance that arises in the nadir regime, and analyzes the operator's rank and identifiability. The resulting structure underlies the SAR comparison formalized later in Section~\ref{sec:sar} and also parallels optical synthetic-aperture imaging
	\cite{curlander1991synthetic, jakowatz_sar_imaging, carrara_spotlight_sar}.
	
	\subsection{Block Structure of the Composite Operator}
	
	Applying the temporal differencing operator to the projected view sequence gives
	\[
	(AP\theta)_{\rho,k}
	= (P_{k+1}\theta)_\rho - (P_k\theta)_\rho.
	\]
	Thus
	\[
	AP =
	\begin{pmatrix}
		P_2 - P_1 \\
		P_3 - P_2 \\
		\vdots \\
		P_{N_t} - P_{N_t-1}
	\end{pmatrix}
	\in \mathbb{R}^{N_p(N_t-1)\times N_g}.
	\]
	Hence the event stream measures temporal differences of successive spatial projections; Section~\ref{sec:sar} relates this difference structure to classical synthetic-aperture systems
	\cite{curlander1991synthetic}.
	
	\subsection{Synthetic-Aperture Interpretation}
	
	Event cameras do not provide absolute image-brightness measurements; instead,
	after temporal binning they constrain changes between the projected views
	$\{P_k\theta\}_{k=1}^{N_t}$ of the fixed ground log--radiance field
	$\theta \in \mathbb{R}^{N_g}$. Reconstruction therefore depends critically on
	how the projection operators $\{P_k\}$ evolve as the sensing platform moves.
	
	We consider a near-nadir flight regime in which the camera undergoes primarily
	horizontal motion with small attitude and altitude variations. In this setting,
	the geometric effect of platform motion is that each pixel footprint moves
	laterally across the ground with limited deformation. Thus the operators
	$\{P_k\}$ are expected to behave approximately as translated versions of a
	common kernel-weighted footprint operator.
	
	This behavior is the structural basis for the SAR analogy formalized in
	Section~\ref{sec:sar}. Each measurement observes the same scene from a shifted
	viewpoint, and the collection of these shifted observations encodes aperture
	diversity \cite{curlander1991synthetic,jakowatz_sar_imaging,carrara_spotlight_sar}.
	The following result makes this approximate shift-structure precise.
	
	For nonzero vectors $u,v$, let $\angle(u,v)$ denote the angle between them, and
	let $\|\cdot\|_F$ denote the Frobenius norm on matrices. For each pixel $\rho$ 
	and time $t_k$, let $\Omega_{\rho,k}\subset\mathbb{R}^2$
	denote the horizontal footprint of that pixel on the ground. The associated
	discrete neighborhood is given by 
	$\mathcal N_{\rho,k}$, and define
	\[
	\operatorname{diam}(\Omega_{\rho,k})
	=
	\sup_{x,y\in\Omega_{\rho,k}}\|x-y\|_2.
	\]
	
	\paragraph{Near-nadir regime.}
	We assume a near-nadir acquisition regime characterized by a small parameter
	$\varepsilon>0$ such that, for all relevant pixels $\rho$ and time indices $k$:
	
	\begin{enumerate}[(i)]
		\item \textbf{Viewing geometry.}
		The viewing directions remain close to nadir. Writing
$\hat d_{\rho,k}^{\mathrm{world}} := R_k^\top d_\rho/\|d_\rho\|_2$ for the
ray direction $d_\rho$ expressed in the world frame at time $t_k$:
		\[
		\angle(\hat d_{\rho,k}^{\mathrm{world}}, \pm e_3)\le \varepsilon,
		\qquad
		\|R_k - R_1\|_F \le \varepsilon.
		\]
		
		with the sign in $\pm e_3$ fixed by the nadir-pointing convention.

\item \textbf{Platform motion.}
		The platform motion is predominantly horizontal:
		\[
		C_k = (\xi_k,h_k), \qquad |h_k - h_1| \le \varepsilon,
		\]
		and we define the horizontal displacement
		\[
		\delta_k = \xi_k - \xi_1.
		\]
		
		\item \textbf{Terrain regularity.}
		Over each footprint domain $\Omega_{\rho,k}\subset\mathbb{R}^2$, the terrain can be locally represented as
		\[
		z = \zeta(x), \quad x \in \Omega_{\rho,k},
		\]
		with $\zeta \in C^2(\Omega_{\rho,k})$ satisfying
		\[
		\|\nabla \zeta\|_{L^\infty(\Omega_{\rho,k})} \le \varepsilon,
		\qquad
		\operatorname{diam}(\Omega_{\rho,k})
		\|D^2 \zeta\|_{L^\infty(\Omega_{\rho,k})} \le \varepsilon.
		\]
	\end{enumerate}
	
	Finally, define the footprint-center mismatch
	\[
	\Delta
	=
	\max_{\rho,k}
	\|X_{\rho,k} - (X_{\rho,1} + \delta_k)\|_2.
	\]
	
	We also assume that the footprint kernel $\kappa$ satisfies the hypotheses
	stated in Section~\ref{sec:projection}: it is Lipschitz, bounded, compactly
	supported, and its normalization denominator is uniformly bounded below. Let
	\[
	M:=\max_{\rho,k}|\mathcal N_{\rho,k}|
	\]
	denote the maximum footprint cardinality. We assume, as is natural for a
	localized footprint discretization, that the overlap multiplicity of the
	footprints is uniformly bounded: there exists $Q>0$ such that each ground index
	belongs to at most $Q$ footprint neighborhoods for a fixed time level.

The near-nadir regime (i)--(iii) above is the physical scenario in which the
footprint-center mismatch $\Delta$ is expected to be small; the theorem
below is stated, and proved, directly in terms of $\Delta$ itself, without
further use of (i)--(iii).
	
	\begin{theorem}[Approximate translation invariance of $P_k$]
		\label{thm:nadir}
		Given the footprint-center mismatch $\Delta$ defined above, there exists a family of linear operators
		\[
		\widetilde P_k \in \mathbb{R}^{N_p \times N_g}
		\]
		and a constant $C>0$, independent of $k$ and $\theta$, such that for all
		$\theta \in \mathbb{R}^{N_g}$,
		\[
		P_k \theta = \widetilde P_k \theta + r_k(\theta),
		\qquad
		\|r_k(\theta)\|_2 \le C \Delta \|\theta\|_2,
		\]
		where:
		\begin{itemize}
			\item $\widetilde P_k$ is obtained by translating the reference operator $P_1$
			by the horizontal displacement $\delta_k$,
			\item $\Delta = \max_{\rho,k} \|X_{\rho,k} - (X_{\rho,1} + \delta_k)\|_2$
			measures the footprint-center mismatch.
		\end{itemize}
	\end{theorem}
	
	\begin{proof}
		Recall that $P_k$ is defined rowwise by
		\[
		(P_k\theta)_\rho
		=
		\sum_{i\in\mathcal N_{\rho,k}} w_{\rho,k,i}\theta_i,
		\quad 
		\mbox{where} \quad 
		w_{\rho,k,i}
		=
		\frac{\kappa(X_i-X_{\rho,k})}
		{\sum_{j\in\mathcal N_{\rho,k}}\kappa(X_j-X_{\rho,k})}.
		\]
		
		For each $k$, define the translated reference footprint center
		\[
		\widetilde X_{\rho,k}:=X_{\rho,1}+\delta_k.
		\]
		Let $\widetilde{\mathcal N}_{\rho,k}$ be the discrete support of the kernel
		centered at $\widetilde X_{\rho,k}$, and define
		\[
		\widetilde w_{\rho,k,i}
		=
		\frac{\kappa(X_i-\widetilde X_{\rho,k})}
		{\sum_{j\in\widetilde{\mathcal N}_{\rho,k}}
			\kappa(X_j-\widetilde X_{\rho,k})},
		\qquad i\in \widetilde{\mathcal N}_{\rho,k}.
		\]
		The translated reference operator $\widetilde P_k$ is given by
		\[
		(\widetilde P_k\theta)_\rho
		=
		\sum_{i\in\widetilde{\mathcal N}_{\rho,k}}
		\widetilde w_{\rho,k,i}\theta_i.
		\]
		By construction, $\widetilde P_k$ is obtained by translating the reference
		kernel-weighted footprint operator at time $t_1$ by the horizontal displacement
		$\delta_k$.
		
		We estimate $P_k-\widetilde P_k$. Since
		\[
		\|X_{\rho,k}-\widetilde X_{\rho,k}\|_2
		\le \Delta,
		\]
		it remains to quantify how the normalized weights change under perturbations
		of the footprint center.
		
		For two admissible centers $X,Y$, define
		\[
		S(X):=\sum_{j\in\mathcal N(X)}\kappa(X_j-X),
		\qquad
		w(X,i):=\frac{\kappa(X_i-X)}{S(X)},
		\]
		where $\mathcal N(X)$ denotes the discrete support of the kernel centered at
		$X$. Extend $w(X,\cdot)$ by zero outside $\mathcal N(X)$.
		
		Using the assumptions that $\kappa$ is Lipschitz with constant $L_\kappa$,
		bounded by $\kappa_{\max}$, and that $S(X)\ge m>0$ for all admissible centers,
		we obtain, on any common finite support of cardinality at most $M$,
		\[
		\begin{aligned}
			|w(X,i)-w(Y,i)|
			&\le
			\frac{|\kappa(X_i-X)-\kappa(X_i-Y)|}{S(X)}
			+
			\kappa_{\max}
			\left|
			\frac{1}{S(X)}-\frac{1}{S(Y)}
			\right|  \\
			&\le
			\frac{L_\kappa}{m}\|X-Y\|_2
			+
			\frac{\kappa_{\max}}{m^2}
			|S(X)-S(Y)| .
		\end{aligned}
		\]
		Moreover,
		\[
		|S(X)-S(Y)|
		\le
		\sum_j |\kappa(X_j-X)-\kappa(X_j-Y)|
		\le
		M L_\kappa \|X-Y\|_2 .
		\]
		Therefore,
		\[
		|w(X,i)-w(Y,i)|
		\le
		\left(
		\frac{L_\kappa}{m}
		+
		\frac{\kappa_{\max}M L_\kappa}{m^2}
		\right)
		\|X-Y\|_2 .
		\]
		Applying this estimate with $X=X_{\rho,k}$ and
		$Y=\widetilde X_{\rho,k}$ gives
		\[
		|w_{\rho,k,i}-\widetilde w_{\rho,k,i}|
		\le
		L_w\Delta,
		\quad 
		\mbox{where} \quad
		L_w :=
		\frac{L_\kappa}{m}
		+
		\frac{\kappa_{\max}M L_\kappa}{m^2}.
		\]
		We now compare the corresponding rows. Extend both weight vectors by zero
		outside their supports and set
		\[
		S_{\rho,k}:=\mathcal N_{\rho,k}\cup \widetilde{\mathcal N}_{\rho,k}.
		\]
		Since each footprint contains at most $M$ grid indices, we have
		$|S_{\rho,k}|\le 2M$. Hence
		\[
		\begin{aligned}
			|(P_k\theta)_\rho-(\widetilde P_k\theta)_\rho|
			&\le
			\sum_{i\in S_{\rho,k}}
			|w_{\rho,k,i}-\widetilde w_{\rho,k,i}|\,|\theta_i|  \\
			&\le
			L_w\Delta
			\sum_{i\in S_{\rho,k}}|\theta_i|  \\
			&\le
			L_w\Delta\,\sqrt{2M}
			\left(
			\sum_{i\in S_{\rho,k}}|\theta_i|^2
			\right)^{1/2},
		\end{aligned}
		\]
		where the last step follows from Cauchy--Schwarz.
		
		Squaring and summing over $\rho$ yields
		\[
		\|P_k\theta-\widetilde P_k\theta\|_2^2
		\le
		2M L_w^2 \Delta^2
		\sum_\rho \sum_{i\in S_{\rho,k}} |\theta_i|^2 .
		\]
		By the bounded overlap assumption, each ground index $i$ appears in at most
		$Q$ sets $S_{\rho,k}$ for fixed $k$. Therefore,
		\[
		\sum_\rho \sum_{i\in S_{\rho,k}} |\theta_i|^2
		\le
		Q\|\theta\|_2^2.
		\]
		Consequently,
		\[
		\|P_k\theta-\widetilde P_k\theta\|_2
		\le
		\sqrt{2MQ}\,L_w\,\Delta\,\|\theta\|_2.
		\]
		
		Finally define
		\[
		r_k(\theta):=(P_k-\widetilde P_k)\theta.
		\]
		Then
		\[
		P_k\theta=\widetilde P_k\theta+r_k(\theta),
		\qquad
		\|r_k(\theta)\|_2\le C\Delta\|\theta\|_2,
		\]
		with $C=\sqrt{2MQ}\,L_w$, which proves the result.
	\end{proof}

	\section{Ill-Posedness of the Forward Model and the Role of Regularization}
	\label{sec:illposed}
	
	The forward model for event-based synthetic-aperture log--radiance reconstruction is
	\[
	AP\theta = b + \eta,
	\]
	where \(A\) is the temporal differencing operator from Section~\ref{sec:temp_event_model}
	and \(P\) is the spatial projection operator from Section~\ref{sec:projection}.
	Although linear, the operator \(AP\) is fundamentally ill-posed due to the
	combined effects of temporal differencing and spatial averaging. These features
	mirror conditioning challenges in classical inverse problems
	\cite{tarantola_inv_prob, engl1996regularization} and synthetic-aperture systems
	\cite{curlander1991synthetic, carrara_spotlight_sar}.
	
	This section characterizes these sources of ill-posedness and motivates the
	regularized formulation used for reconstruction. All operators are analyzed
	with respect to the Euclidean inner product on the discretized log--radiance space.
	
	\subsection{Rank, Nullspace, and Ill-Conditioning of \texorpdfstring{$AP$}{AP}}
	\label{sec:rank_conditioning}
	
	Ignoring the additive noise term, the forward model reads
	\[
	b = AP\theta, 
	\qquad 
	(AP\theta)_{\rho,k} = \bigl((P_{k+1}-P_k)\theta\bigr)_\rho.
	\]
	Thus identifiability and stability depend on the structure of the projection
	sequence $\{P_k\}$ and the temporal differencing operator $A$.
	Ill-posedness arises from two mechanisms: a geometric nullspace induced by
	insufficient viewpoint diversity and poor conditioning of the operators.
	
	\paragraph{Geometric nullspace and synthetic-aperture diversity.}
	
	By construction,
	\[
	AP\theta = 0 
	\quad \Longleftrightarrow \quad
	P_{k+1}\theta = P_k\theta \quad \forall k.
	\]
	Thus any log--radiance field producing identical projections across all
	viewpoints lies in $\ker(AP)$. The size of this nullspace depends on the
	variation in $\{P_k\}$.
	
	Under the near-nadir regime of Theorem~\ref{thm:nadir}, the projection operators satisfy
	\[
	P_k \theta = \widetilde P_k \theta + r_k(\theta),
	\qquad
	\|r_k(\theta)\|_2 \le C\Delta \|\theta\|_2,
	\]
	where $\widetilde P_k$ is obtained by translating a reference operator.
	Thus, to leading order, the projections behave as translated versions of a
	common footprint operator, and identifiability depends on the sequence of
	displacements $\{\delta_k\}$.
	
	If the displacements $\delta_k$ span a sufficiently large range, the operators
	$\{\widetilde P_k\}$ probe distinct regions of the scene, reducing the nullspace.
	Conversely, if $\delta_k$ is nearly constant, then $P_{k+1}\approx P_k$ and
	$AP$ has a large nullspace. This behavior reflects the synthetic-aperture principle:
	spatial diversity is generated by viewpoint translation
	\cite{curlander1991synthetic,jakowatz_sar_imaging}.
	
	\paragraph{Degenerate camera geometries.}
	
	The geometric diversity of $\{P_k\}$ collapses under several degeneracies:
	\begin{itemize}
		\item \textbf{Pure rotation:} if $C_k$ is constant, then $P_k = P_{k'}$.
		\item \textbf{Near-zero baseline:} insufficient lateral motion yields
		$P_{k+1}\approx P_k$ and highly correlated measurements.
		\item \textbf{Textureless regions:} if $\theta$ is nearly constant, shifted
footprints produce nearly identical event signals.
	\end{itemize}
	In the first two cases, the effective rank of $AP$ itself is reduced, leading to synthetic-aperture
	collapse and loss of identifiability \cite{carrara_spotlight_sar}. In the
third case, $AP$ remains well-determined, but the resulting event signal is
weak and easily dominated by noise, reducing practical observability.
	
	\paragraph{Algebraic conditioning of $A$, $P$, and $AP$.}
	
	Even when the geometric nullspace is small, the operators $A$ and $P$ are
	individually ill-conditioned due to their spectral behavior.
	
	The temporal differencing operator $A$ annihilates constant signals:
	\[
	A(y + c\mathbf{1}) = Ay,
	\]
	so $\mathbf{1} \in \ker(A)$.

In particular, since the footprint weights are normalized
($\sum_{i}w_{\rho,k,i}=1$, Section~\ref{sec:projection}), $P_k\mathbf{1}_{N_g}=\mathbf{1}_{N_p}$
for every $k$, so $AP(\theta+\alpha\mathbf{1}_{N_g})=AP\theta$ for any scalar
$\alpha$: the ground-domain log--radiance field $\theta$ is identifiable
from event data only up to a global additive constant, unless an
anchoring condition (e.g.\ $\mathbf{1}^\top\theta=0$) or a gauge-fixing
regularizer is imposed.

More generally, $A$ suppresses slowly varying
	temporal components and amplifies rapidly varying ones. For standard
	finite-difference constructions, its singular values vanish at zero frequency
	and increase with temporal frequency \cite{leveque_finite_difference}. Thus $A$
	acts as a temporal high-pass operator.
	
	The projection operator $P$ acts as a spatial averaging operator. Each row
	represents a localized average of $\theta$, so $P$ attenuates high spatial
	frequencies while preserving low-frequency structure. Consequently, $P$ behaves
	as a spatial low-pass operator, leading to decay of singular values at high
	spatial frequencies \cite{tarantola_inv_prob, engl1996regularization}.
	
	The composite operator $AP$ therefore combines spatial smoothing with temporal
	differencing. This suppresses fine spatial detail while amplifying temporal
	variations, producing a wide spread in singular values. In particular,
	\[
	\sigma_{\max}(AP) \le \sigma_{\max}(A)\,\sigma_{\max}(P),
	\]
	the norm of the composite operator is therefore controlled by the norms of
	$A$ and $P$, while small singular values can arise from temporal
differencing, spatial averaging, or geometric lack of diversity.
	Since $AP$ has a nontrivial nullspace (the geometric nullspace discussed
above), $\sigma_{\min}(AP)=0$ exactly; conditioning is therefore more
meaningfully quantified by
	\[
	\mathrm{cond}^+(AP)
	=
	\frac{\sigma_{\max}(AP)}{\sigma_{\min}^+(AP)},
	\]
	where $\sigma_{\min}^+(AP)$ is the smallest positive singular value, i.e.\
conditioning restricted to the orthogonal complement of the nullspace; the
combined smoothing and differencing effects above make this quantity
large, indicating severe ill-conditioning on the identifiable subspace.
	
	\subsection{Regularized Inversion}
	\label{sec:map}
	
	The linearized forward model takes the form
	\[
	AP\theta = b + \eta,
	\]
	where $\eta \in \mathbb{R}^{N_p(N_t-1)}$ is the stacked noise vector introduced
	in Section~\ref{sec:temp_event_model}.
	
	To obtain a stable reconstruction, we combine the data misfit with a
	regularization term that encodes prior information about the spatial structure
	of the scene. This leads to the estimator
	\[
	\widehat{\theta}
	=
	\arg\min_{\theta}
	\left[
	\frac12\|AP\theta - b\|_2^2 + R(\theta)
	\right],
	\]
	where $R(\theta)$ is a regularization functional, such as Total Variation or
quadratic ($L^2$) regularization. Total Variation is implemented via a
smoothed (differentiable) approximation,
$\mathrm{TV}_\epsilon(\theta)=\sum_i\sqrt{|\nabla_h\theta_i|^2+\epsilon^2}$,
compatible with the L--BFGS solver and analytic gradients used in Section~\ref{sec:computational}.
	
	The role of the regularizer is to compensate for both the nullspace and the
	poor conditioning of $AP$. In particular, it suppresses unstable components
	associated with small singular values while promoting physically plausible
	solutions with controlled spatial variation. In the examples below, we
	consider both Total Variation and quadratic ($L^2$) regularization, and the
framework accommodates alternative
	priors depending on the application regime
	\cite{tarantola_inv_prob, engl1996regularization, tikhonov1977solutions}.
	
	\section{Relationship to Synthetic-Aperture Radar}
	\label{sec:sar}
	
	Before turning to computational results, we formalize the SAR analogy invoked in the preceding sections. The comparison is between classical SAR forward operators and the event-based composite model
	\[
	AP\theta = b + \eta,
	\]
	whose rows encode aperture-indexed differences of projected scene measurements. Although the sensing modalities differ—optical radiance versus radar backscatter—the underlying mathematical structure is similar: a moving platform collects data whose measurements at each position depend on how a spatially smoothed scene changes across the aperture. This section describes these connections, emphasizing how the geometry and conditioning of the composite operator $AP$ parallel those of classical SAR operators.
	
	\subsection{Comparison of Forward Models}
	
	In SAR imaging, the measurement at aperture position $k$ is typically modeled as a convolution of the surface reflectivity with a range–azimuth kernel determined by the transmitted waveform and platform geometry \cite{curlander1991synthetic}. Successive radar pulses produce shifted versions of this kernel as the platform moves along-track, and inversion relies on spanning a sufficiently large baseline to ensure coverage of spatial frequencies in the aperture domain \cite{carrara_spotlight_sar}.
	
	Our event-based optical model exhibits an analogous structure at the level of the composite forward operator. The event measurements satisfy
	\[
	(AP\theta)_{\rho,k} = (P_{k+1}\theta - P_k\theta)_\rho.
	\]
	Under the near-nadir regime (Theorem~\ref{thm:nadir}), the projection operators admit the approximation
	\[
	P_k \theta = \widetilde P_k \theta + r_k(\theta),
	\qquad
	\|r_k(\theta)\|_2 \le C\Delta \|\theta\|_2,
	\]
	where $\widetilde P_k$ is obtained by translating a reference footprint operator. Substituting this into the forward model yields
	\[
	(AP\theta)_{\rho,k}
	=
	\bigl((\widetilde P_{k+1} - \widetilde P_k)\theta\bigr)_\rho
	+
	\bigl(r_{k+1}(\theta) - r_k(\theta)\bigr)_\rho.
	\]
	Thus each event measurement is a difference of shifted, spatially smoothed scene samples across adjacent aperture positions, up to a controlled perturbation.
	
	This shift-difference structure is the direct analogue of aperture-indexed SAR measurements: the comparison arises from how the aperture encodes differences across translated views, rather than from an exact convolutional representation. In both systems, sufficient platform motion is required to generate diversity in the rows of the forward operator.
	
	\subsection{Ill-Conditioning and Aperture Diversity}
	
	As in SAR, the invertibility of the composite forward model depends critically on the platform's geometric baseline. If the lateral displacement $\delta_k$ is too small, then adjacent operators $P_k$ become nearly identical, so the corresponding rows of $AP$ lose diversity and many spatial modes of the scene become weakly observable or unobservable.
	
	In SAR, such degeneracy corresponds to poor angular diversity or insufficient Doppler bandwidth \cite{curlander1991synthetic,jakowatz_sar_imaging}. In the present setting, it manifests as an inability to resolve spatial variations in $\theta$ from the shifted-difference measurements generated by $AP$.
	
	Moreover, both SAR and the event-based model exhibit strong low-pass characteristics. The projection operator $P$ smooths $\theta$ spatially through the footprint kernel $\kappa$, analogous to the point-spread function in SAR imaging. In both cases, deblurring amplifies instability and leads to severe ill-conditioning \cite{carrara_spotlight_sar,engl1996regularization}.
	
	The temporal differencing operator $A$ further exacerbates this by eliminating all temporally constant modes. Consequently, the composite operator $AP$ combines spatial low-pass filtering with temporal high-pass filtering, leading to a broad spread in singular values. This mirrors the conditioning challenges encountered in multi-aperture SAR systems \cite{jakowatz_sar_imaging}. The SAR analogy is therefore most faithful at the level of the composite operator $AP$, rather than the projection family $\{P_k\}$ alone.
	
	\subsection{Event Cameras vs.\ SAR Sensors}
	
	Despite their differences, both sensing modalities share key structural features:
	
	\begin{itemize}
		\item \textbf{Shift-difference forward operators.}  
		Both systems generate aperture-indexed measurements from a moving platform observing a common scene through shifted, geometry-dependent kernels. In the event-camera case, the additional temporal differencing induced by the sensor produces a shift-difference operator.
		
		\item \textbf{Dependence on aperture diversity.}  
		Both require platform motion to span sufficient spatial frequencies; reduced baseline leads to degraded conditioning of the measurement operator.
		
		\item \textbf{Intrinsic ill-posedness.}  
		SAR inherits ill-conditioning from convolutional blur and limited aperture; the event-based model combines spatial smoothing with temporal differencing, yielding a structurally similar inverse problem.
	\end{itemize}
	
	The primary distinction lies in the sensing mechanism. SAR actively illuminates the scene and measures coherent returns, capturing phase and Doppler information. Event cameras, by contrast, passively record asynchronous changes in log--radiance, effectively measuring temporal gradients of intensity. Despite these differences, both modalities produce aperture-driven measurement operators whose inversion depends on baseline diversity, blur, and differencing structure. The closest correspondence is therefore between SAR forward maps and the composite operator $AP$.
	
	\section{Computational Results}
	\label{sec:computational}
	
	This section presents computational results obtained using a direct
	implementation of the forward model and regularized estimator developed in
	Sections~\ref{sec:forward_model}--\ref{sec:map}. Rather than relying on idealized
	operators, all experiments use explicit constructions of the temporal
	differencing operator, motion-induced alignment, and projection matrices,
	closely reflecting how the method would be deployed in practice.
	
	The key distinction from the dynamic inverse formulation of \cite{antil2024drnd}
	is the reconstruction variable. The method in \cite{antil2024drnd} recovers a
	latent time-dependent image sequence, whereas the present work solves directly
	for the static ground-domain log--radiance field $\theta$ under the composite
	model $AP\theta=b+\eta$. The synthetic-aperture reconstructions shown here are
	therefore specific to the $\theta$-based formulation, while the dynamic
	reconstruction method of \cite{antil2024drnd} is included only as a baseline
	in the real-data comparison.
	
	\subsection{Computational Pipeline}
	\label{sec:pipeline}
	
	The computational pipeline consists of four main stages:
	(i) event loading and dynamic temporal binning,
	(ii) motion-induced alignment of event coordinates,
	(iii) construction of the projection operators $\{P_k\}$,
	and (iv) numerical solution of the regularized optimization problem.
	Each stage corresponds directly to a component of the analytical model.
	
	Event loading and temporal binning are performed jointly, rather than using
	fixed-duration bins.  Instead, events are dynamically grouped so that each bin
	corresponds to an estimated one-pixel horizontal shift of the underlying scene
	relative to the previous bin.  This binning strategy is designed to enforce the
	translated-footprint geometry assumed in the forward model and thereby produce
	a tractable approximation of the composite operator $AP$.
	
	Specifically, events are accumulated over candidate temporal windows and
	projected onto the sensor grid. For each candidate bin boundary, a small set of
	integer horizontal shifts is evaluated by applying the shift to the accumulated
	events and computing a contrast metric on the resulting image. A bin is closed
	when the shift that maximizes contrast corresponds to a one-pixel translation
	relative to the previous bin. This procedure yields a sequence of bins whose
	effective motion is approximately uniform and consistent with the
	translated-footprint model analyzed in Section~\ref{sec:forward}.
	
	This one-pixel shift binning also provides a direct and computationally
	efficient approximation of the projection operators $\{P_k\}$, and hence of the
	composite forward map $AP$. Under near-nadir motion, Theorem~\ref{thm:nadir}
	shows that each $P_k$ is well approximated by a translated convolution
	operator. By dynamically selecting bins such that successive event
	accumulations correspond to estimated integer-pixel translations, the action of
	$P_k$ is discretized as a sparse shift operator on the ground-domain grid.
	
	In implementation, this approximation replaces continuous footprint motion with
	piecewise-constant integer translations, so that each $P_k$ acts as a sparse
	shift operator on $\theta$ while masking enforces a fixed reconstruction
	domain. Applying the temporal differencing operator $A$ to this approximated
	projection sequence yields a tractable realization of the composite operator
	$AP$. The resulting operators capture the dominant geometric effect of platform
	motion while avoiding the computational cost and modeling uncertainty
	associated with subpixel footprint interpolation.
	
	\subsection{Synthetic Reconstruction Experiments}
	\label{sec:synthetic}
	
	Synthetic experiments are designed to isolate and validate the
	shift-difference behavior of the composite forward operator under controlled,
	simulator-generated conditions.  A ground-truth log--radiance map
	$\theta^\star \in \mathbb{R}^{240\times180}$ is defined on a fixed ground-domain
	grid.  The platform's horizontal position $\xi_k$ is set to a pure
	rightward translation, with $h_k\equiv h_1$ and $R_k\equiv R_1$:
	\[
	\xi_k = (k-1)\ \text{pixels}, \qquad k = 1,\dots,101.
	\]
	By the definition $\delta_k:=\xi_k-\xi_1$ from Theorem~\ref{thm:nadir}, this
makes $P_k\theta^\star$ exactly the rightward translation of $\theta^\star$ by
$\delta_k=(k-1)$ pixels, with $P_1$ as the unshifted reference frame and
	no additional geometric deformation.
	
	Synthetic event measurements are generated using the open-source
	v2e event camera simulator \cite{v2e},
	which converts image sequences into asynchronous event streams according to a
	standard contrast-threshold event model. The translated projected views induced
	by $\theta^\star$ are provided to the simulator, and events are generated from
	successive frames corresponding to known one-pixel horizontal shifts.
	
	For consistency with the analytical model, the resulting event stream is then
	binned and accumulated to form signed polarity counts satisfying
	\[
	P_{k+1}\theta^\star - P_k\theta^\star = b_k + \eta_k, \qquad k = 1,\dots,100,
	\]
	where $b_k,\eta_k \in \mathbb{R}^{N_p}$ denote the pixel-stacked measurement and
	error vectors for the $k$th temporal bin, so that the stacked measurements
	satisfy the composite forward model $AP\theta^\star = b + \eta$.
	All noise and thresholding options were kept at their defaults provided by the v2e software; in particular, the nominal ON/OFF contrast thresholds were left at $0.2$ in log-intensity units for both polarities.

	In implementation, the dynamic binning and projection pipeline described in Section~\ref{sec:pipeline} reduces to a known, fixed sequence of one-pixel
	shifts.  Each projection operator $P_k$ is therefore implemented as a sparse
	shift matrix acting on the static log--radiance estimate $\theta$. To enforce a
	consistent reconstruction domain across all shifts, pixels that leave the field
	of view over the motion sequence are masked, ensuring that all operators act on
	a common support.
	
	Reconstruction is performed by minimizing the regularized objective from
	Section~\ref{sec:map}. Optimization is carried out using L--BFGS with analytic
	gradients. Figure~\ref{fig:synthetic_results} shows that the recovered log--radiance
	closely matches the ground truth. These results confirm that, when the
	shift-difference composite model is satisfied exactly, the implemented
	operator faithfully recovers the underlying scene structure from temporal
	differences alone. The comparatively modest PSNR should be interpreted with
	care: the event-only model constrains temporal changes and relative spatial
	contrast more directly than exact absolute log--radiance values. As a result, strong
	structural agreement can coexist with lower pixelwise log--intensity agreement,
	which is consistent with the high SSIM observed here.

	\begin{figure}[h!]
		\centering
		\subfloat{\includegraphics[width=0.31\linewidth]{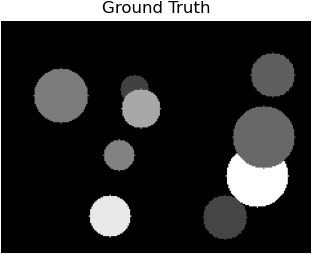}}
		\subfloat{\includegraphics[width=0.31\linewidth]{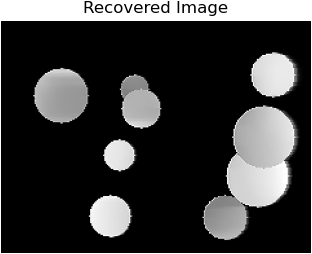}}
		\raisebox{-4.3 pt}{\subfloat{\includegraphics[width=0.356118\linewidth]{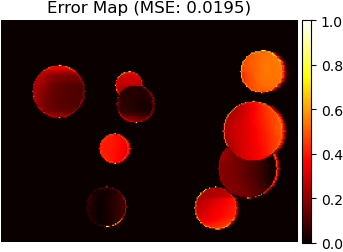}}}
		\caption{
			Synthetic reconstruction experiment.
			\textbf{Left:} ground-truth log--radiance $\theta^\star$.
			\textbf{Center:} reconstructed log--radiance estimate $\widehat{\theta}$.
			\textbf{Right:} squared error map.
			The reconstruction achieves PSNR $=17.11\,\mathrm{dB}$ and
			SSIM $=0.9128$. SSIM remains high even though PSNR is modest, reflecting
			that the event-only model recovers relative spatial structure and contrast
			more reliably than exact absolute log--radiance values.
		}
		\label{fig:synthetic_results}
	\end{figure}

	\subsection{Real Event-Camera Data Under Near-Nadir Motion}
	\label{sec:realdata}
	
	We next apply the proposed reconstruction framework to real event-camera data
	collected by the \emph{Falcon Neuro} neuromorphic payload aboard the International Space Station (ISS) \cite{mcharg2022falcon}.  The dataset consists of
	asynchronous polarity events recorded during nominal ISS operations, with the
	sensor viewing the Earth in a near-nadir configuration.  Platform motion induces lateral translation of the scene across the sensor, providing the aperture
	diversity required for synthetic-aperture reconstruction.
	
	Event loading and temporal binning are performed dynamically, following the
	procedure described in Section~\ref{sec:pipeline}.  Rather than fixing the
	bin duration, events are accumulated until the contrast-maximizing horizontal
	shift between successive accumulations corresponds to an estimated one-pixel
	translation.  This produces a sequence of bins whose effective motion is
	approximately uniform and directly compatible with the translated-footprint forward model.
	
	For each bin, events are aligned by compensating for the estimated shift, mapping all measurements into a common ground-referenced coordinate system.  This alignment step reduces the effect of residual motion within bins and ensures that each projection operator $P_k$ acts on a fixed reconstruction domain. The operators $\{P_k\}$ are then constructed as sparse matrices encoding the association between aligned sensor pixels and ground-domain cells, approximating the translated convolution operators predicted by Theorem~\ref{thm:nadir}.
	
	Reconstruction is performed by minimizing the regularized objective derived in
	Section~\ref{sec:map}. Optimization is carried out using an L--BFGS solver with
	analytic derivatives.
	
	Figure~\ref{fig:synthetic_comparison} shows a synthetic comparison on the same
	one-pixel-shift experiment described in Section~\ref{sec:synthetic}. In that
	figure, the proposed method is compared against FIRENET and E2VID using data
	generated from the known translated projection sequence. Figure~\ref{fig:panama_reconstruction}
	shows the corresponding comparison on real Falcon Neuro data, where the
	dynamic reconstruction of \cite{antil2024drnd} is included as an additional
	baseline because it reconstructs a time-dependent image sequence rather than
	the static scene parameter $\theta$. The synthetic-aperture reconstruction exhibits coherent
	large-scale spatial structure, with clear delineation of scene geometry and
	consistent extended boundaries, despite relying solely on
	event-based temporal measurements. Fine-scale texture and high-frequency detail
	are largely suppressed, consistent with the effects of spatial footprint
	averaging, temporal differencing, and regularization inherent to the
	formulation. In the real-data comparison shown in
	Figure~\ref{fig:panama_reconstruction}, the reconstruction in panel 1 is an
	$L^2$-regularized estimate computed on the native $240\times180$ grid rather
	than the super-resolved ablation setting. In contrast, the dynamic
	reconstruction shows similar large-scale geometry but with increased spatial
	variation and residual artifacts.
	Learning-based methods (FIRENET and E2VID) recover fine-scale texture and local
	contrast, but also introduce granular noise, streaking, and spatially
	inconsistent intensity patterns, particularly in regions of dense event
	activity. Overall, the comparison highlights the ability of the
	synthetic-aperture approach to recover stable, geometrically aligned log--radiance
	structure while trading off high-frequency texture for robustness and spatial
	coherence.
	
	These results demonstrate that dynamic one-pixel binning and shift-based
	approximations of $P_k$ are sufficient to recover meaningful spatial
	information from real spaceborne event data in the native-resolution
	reconstruction setting shown here, supporting the practical applicability of
	the synthetic-aperture framework.
	
	In these experiments, the number of sensor pixels $N_p$ is fixed by the event
	camera, whereas the number of unknown ground-domain cells $N_g$ is a user-chosen
	discretization parameter for the log--radiance field $\theta$. The formulation does
	not require $N_g=N_p$. In the super-resolved ablation study below, we choose
	$N_g \gg N_p$, so that $\theta$ is reconstructed on a grid finer than the
	native image resolution. The corresponding projection operators $P_k$ map each
	sensor pixel to a localized footprint on this finer grid, implemented here with
	a $5\times5$ kernel. In this sense, super-resolution arises from solving for a
	higher-resolution scene representation under the same event measurements, rather
	than from modifying the event sensing model itself.

	\begin{figure}[!h]
		\centering
		{\subfloat{\includegraphics[width=0.24\linewidth]{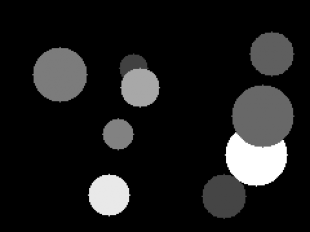}}}
		{\subfloat{\includegraphics[width=0.24\linewidth]{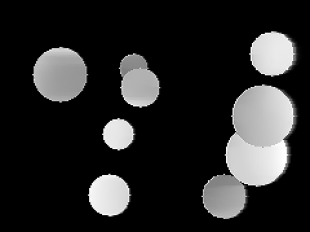}}}
		\subfloat{\includegraphics[width=0.24\linewidth]{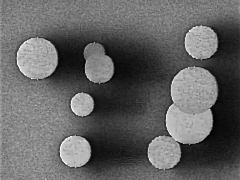}}
		\hspace{.015in}
		\subfloat{\includegraphics[width=0.24\linewidth]{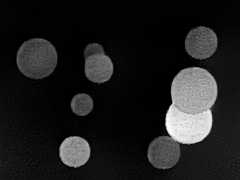}}
		
		\caption{
			Comparison on the same synthetic one-pixel-shift experiment described in
			Section~\ref{sec:synthetic}. From left to right:
			\textbf{1.} Ground truth log--radiance.
			\textbf{2:} Reconstruction obtained using the event-based synthetic-aperture
			(SAR-style) formulation developed in this work.
			\textbf{3:} Reconstruction obtained using the FIRENET method by \cite{Scheerlinck20wacv}.
			\textbf{4:} Reconstruction obtained using the E2VID method by \cite{Rebecq19pami}.}
		\label{fig:synthetic_comparison}
	\end{figure}

	\begin{figure}[!h]
		\centering
		\raisebox{-1.5pt}{\subfloat{\includegraphics[width=0.2465\linewidth]{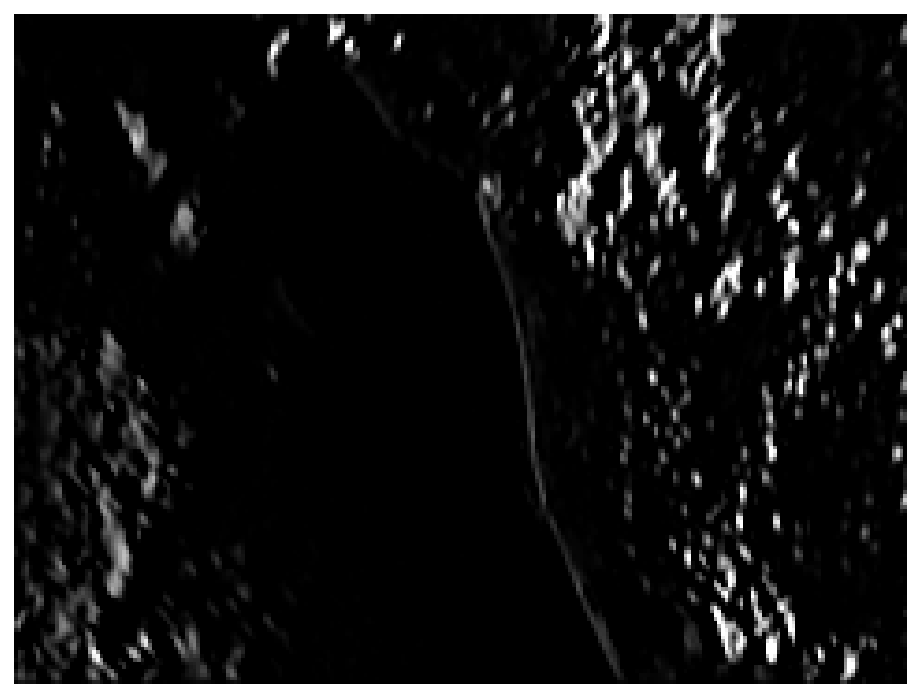}}}
		\subfloat{\includegraphics[width=0.24\linewidth]{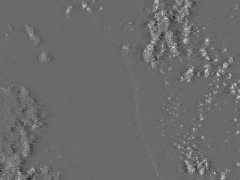}} 
		\hspace{.02in}
		\subfloat{\includegraphics[width=0.2402\linewidth]{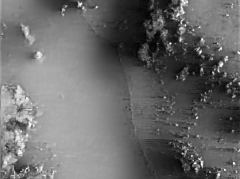}}
		\hspace{.02in}
		\subfloat{\includegraphics[width=0.2402\linewidth]{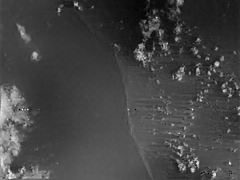}}

		\caption{
			Final reconstructed log--radiance maps from real event-camera data collected
			under near-nadir motion. From left to right: 
			\textbf{1:} Reconstruction obtained using the event-based synthetic-aperture
			(SAR-style) formulation developed in this work, using $L^2$
			regularization on the native $240\times180$ grid.
			\textbf{2:} Reconstruction obtained using the dynamic reconstruction method
			introduced by \cite{antil2024drnd}.
			\textbf{3:} Reconstruction obtained using the FIRENET method by \cite{Scheerlinck20wacv}.
			\textbf{4:} Reconstruction obtained using the E2VID method by \cite{Rebecq19pami}.}
		\label{fig:panama_reconstruction}
	\end{figure}
	
	\subsection{Regularization Ablation Study}
	\label{sec:ablation}
	
	To assess the effect of the regularization choice and parameter magnitude, we
	compare reconstructions obtained with Total Variation (TV) and quadratic
	$L^2$ regularization over a representative range of $\lambda$ values. The
	resulting ablation is shown in Figure~\ref{fig:ablation_study}, with rows
	indexed by $\lambda$ and columns corresponding to the two regularizers. In all
	ablation runs, the reconstruction is super-resolved to $1920\times1440$ and
	uses a $5\times5$ projection kernel.
	
	Across the sweep, small values of $\lambda$ leave the inversion weakly
	regularized, producing reconstructions with greater local variation and reduced
	spatial coherence. As $\lambda$ increases, both models suppress these
	artifacts, but they do so differently: TV better preserves sharp transitions
	and extended boundaries, whereas $L^2$ regularization produces smoother images
	with more diffuse edges. This comparison highlights the tradeoff between edge
	preservation and smoothing, and should be interpreted separately from the
	native-resolution real-data comparison shown in
	Figure~\ref{fig:panama_reconstruction}.
	
	\begin{figure}[H]
		\centering
		\footnotesize
		\setlength{\tabcolsep}{4pt}
		\renewcommand{\arraystretch}{0.95}
		\begin{tabular}{c c c}
			$\lambda$ & TV & $L^2$ \\
			$0.10$ &
			\includegraphics[width=0.23\linewidth]{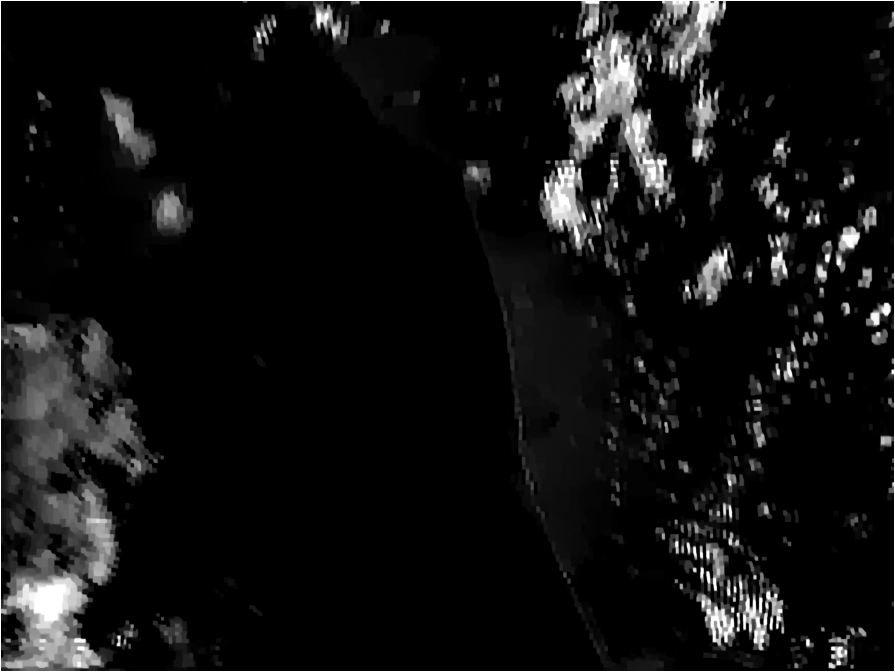} &
			\includegraphics[width=0.23\linewidth]{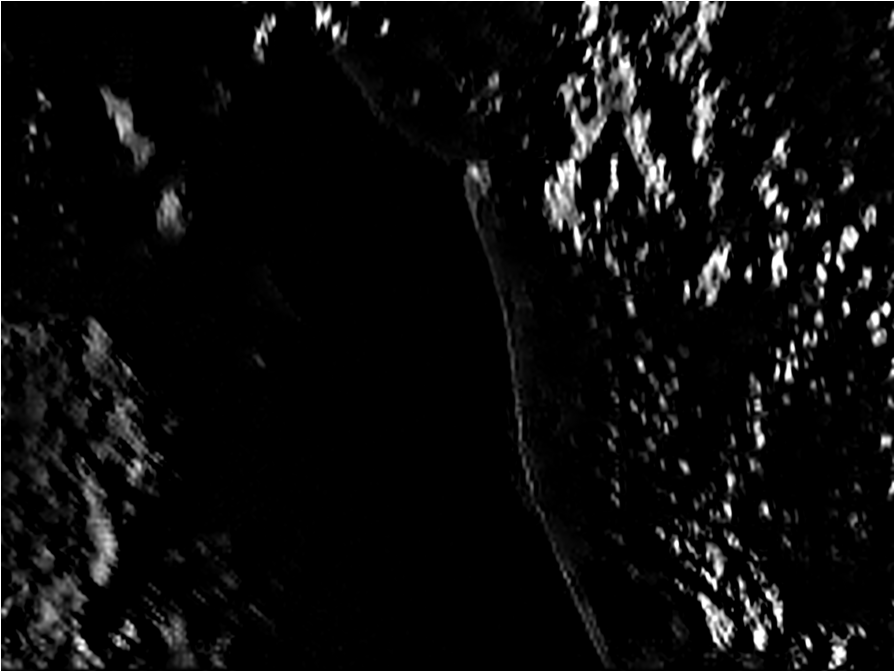} \\
			$0.20$ &
			\includegraphics[width=0.23\linewidth]{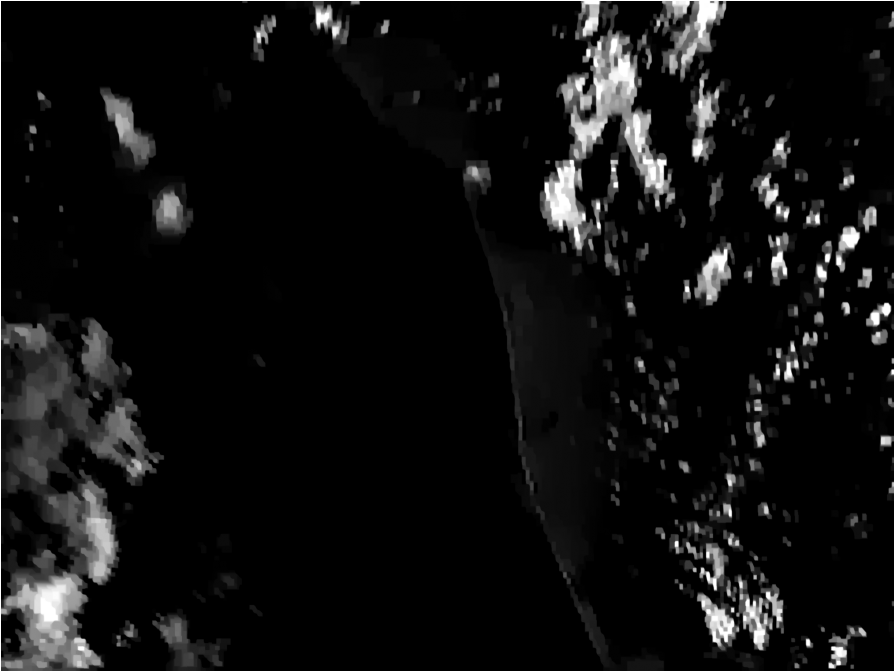} &
			\includegraphics[width=0.23\linewidth]{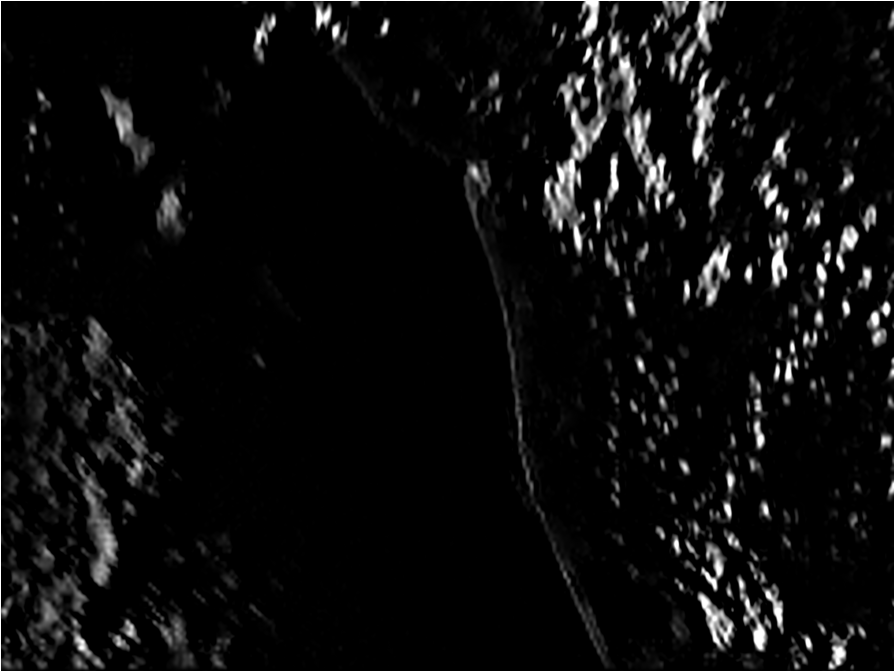} \\
			$0.30$ &
			\includegraphics[width=0.23\linewidth]{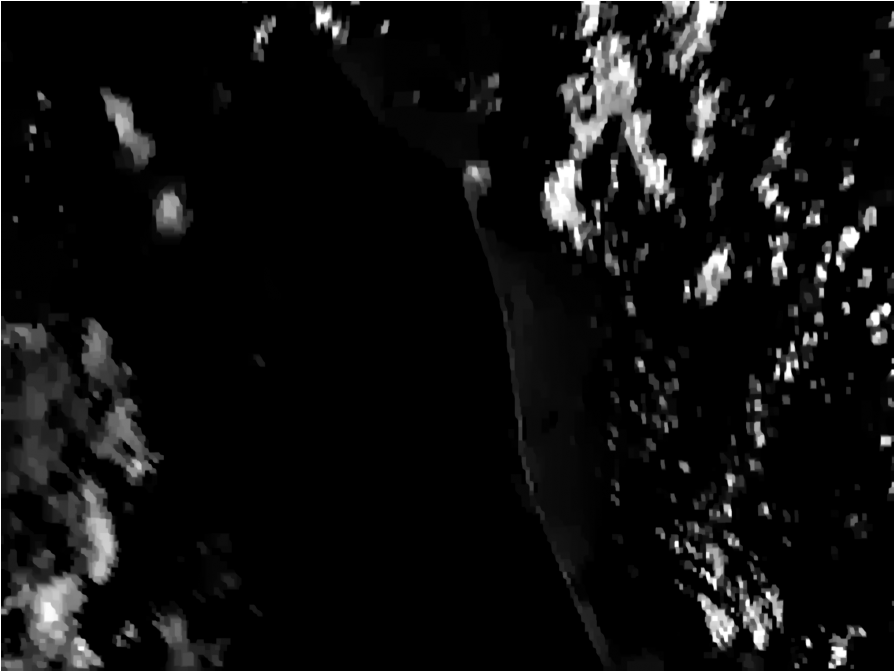} &
			\includegraphics[width=0.23\linewidth]{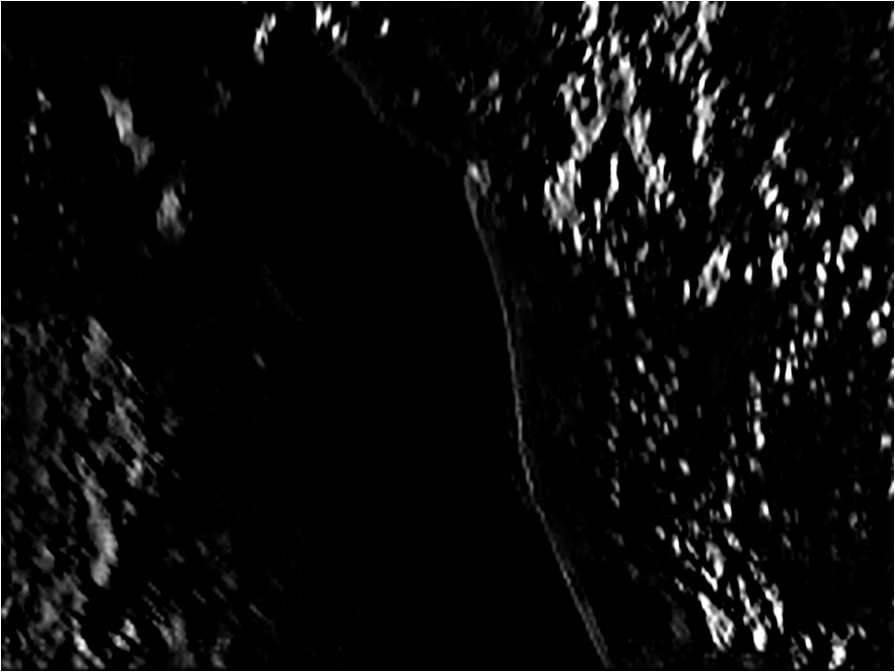} \\
			$0.40$ &
			\includegraphics[width=0.23\linewidth]{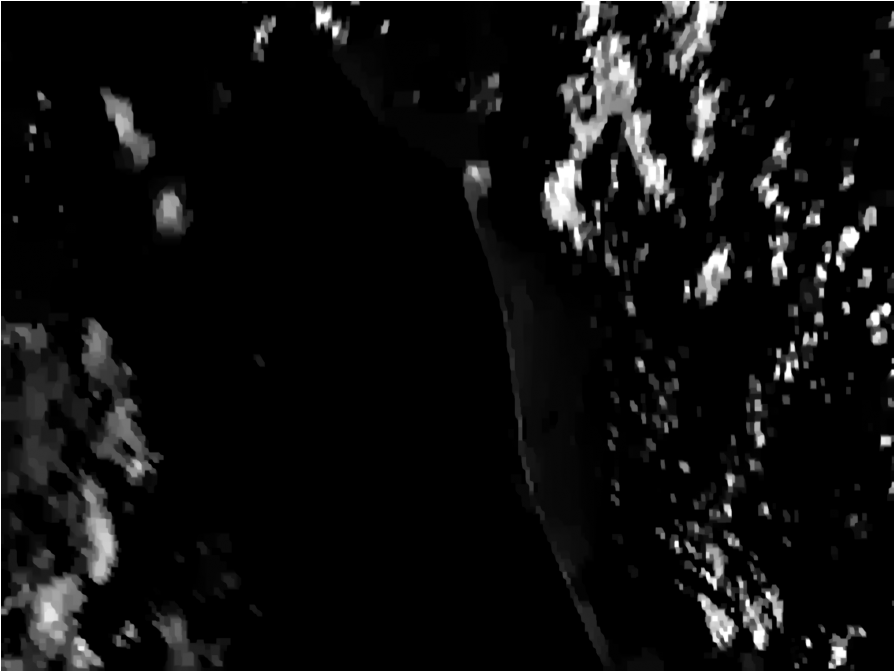} &
			\includegraphics[width=0.23\linewidth]{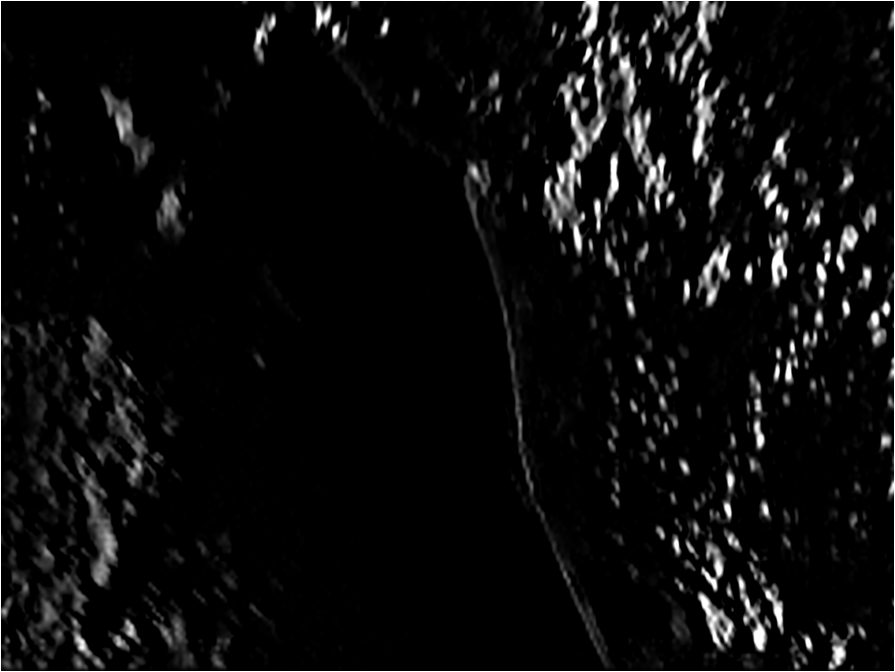} \\
			$0.50$ &
			\includegraphics[width=0.23\linewidth]{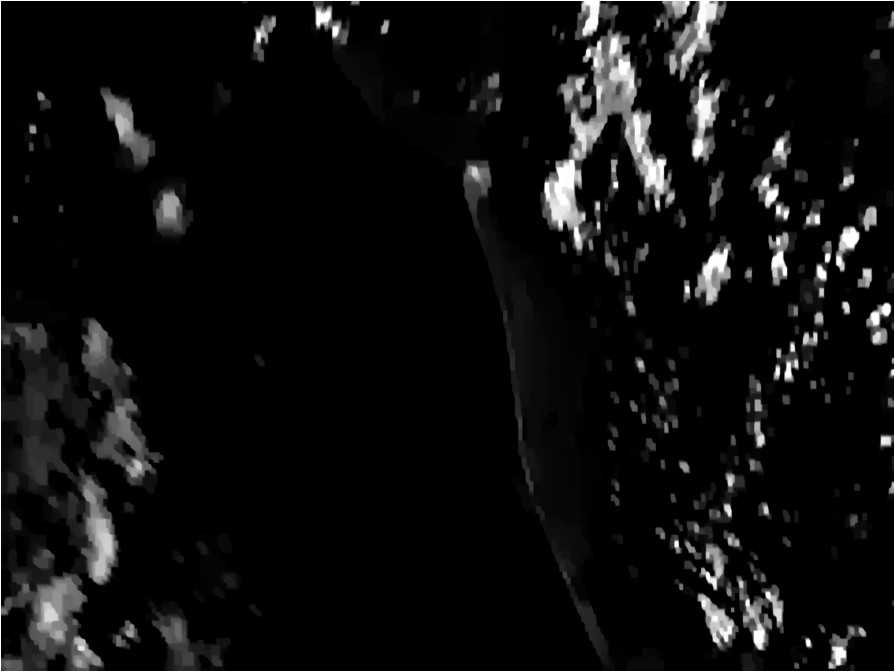} &
			\includegraphics[width=0.23\linewidth]{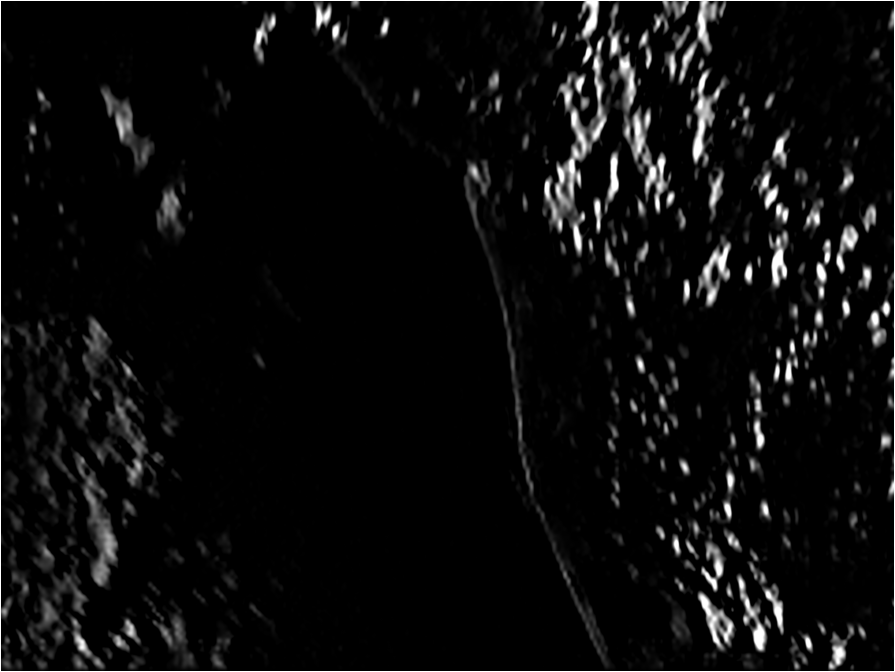} \\
		\end{tabular}
		\caption{
			Regularization ablation study. Rows correspond to the regularization
			parameter $\lambda \in \{0.1,0.2,0.3,0.4,0.5\}$, while columns compare
			TV and quadratic $L^2$ regularization on the same reconstruction task.
			All reconstructions are super-resolved to $1920\times1440$ using a
			$5\times5$ kernel. As $\lambda$ increases, both models become smoother,
			with TV preserving sharper scene boundaries and $L^2$ producing more
			diffuse transitions.
		}
		\label{fig:ablation_study}
	\end{figure}
	
	\section{Conclusion \& Future Work}
	
	\subsection{Conclusion}
	This work developed a synthetic-aperture interpretation of monocular
	event-camera sensing by formulating event measurements as the action of a
	composite linear operator $AP$ on a static ground-domain log--radiance field.
	By explicitly separating temporal differencing from spatial projection, the
	resulting forward model exposes a shift-difference operator analogous to those
	arising in classical synthetic-aperture radar.  Under near-nadir motion, this
	structure yields aperture diversity through lateral platform translation,
	providing a principled explanation for spatial identifiability from purely
	event-based measurements.
	
	We showed that the composite operator $AP$ is inherently ill-conditioned due to
	both temporal high-pass filtering and spatial low-pass averaging. These
	properties motivate a regularized estimator for the event-based setting.
	Theoretical analysis clarified the roles of motion and baseline in
	determining conditioning and reconstruction
	quality.
	
	Computational experiments on both synthetic data and real spaceborne event
	imagery demonstrate that the proposed framework can recover coherent spatial
	structure using only asynchronous brightness-change measurements.
	The results confirm that sufficient lateral motion induces the same type of
	aperture diversity exploited in SAR, while regularization provides stability
	in the presence of severe operator ill-conditioning.
	
	\subsection{Future Work}
	
	The synthetic-aperture formulation developed in this work opens several
	theoretical and practical directions for further investigation.  These
	extensions naturally build on the operator-based perspective and address
	limitations inherent in monocular, event-only sensing. The directions below
	range from immediate extensions of the current framework to broader
	scene-representation questions suggested by the results.
	
	\paragraph{Digital twin through multi-view and multi-sensor extensions.}
	A natural extension of the operator-based formulation developed here is the
	integration of a persistent digital twin as the latent scene representation.
	In this setting, the unknown scene parameter is promoted from a two-dimensional log--radiance map to a three-dimensional, measurement-predictive model that is shared across repeated or distributed observations. Rather than reconstructing
	independent images for each collection, new event measurements are used to test and refine a common scene hypothesis through the forward model
	$b^{(m)} = A^{(m)} P^{(m)} \theta$.
	
	When measurements are acquired from multiple sensors or passes with distinct
	viewing geometries, the framework naturally supports joint estimation over all
	views. Each sensor contributes its own projection operator, while the underlying scene representation remains shared, enabling depth and vertical structure to be constrained through cross-view consistency. Platform motion further enriches the available viewpoints, allowing three-dimensional scene structure to emerge from the collective measurements without requiring explicit correspondence or
	pairwise matching.
	
	\paragraph{Super-resolution.}
	The formulation developed in this work supports super-resolution, as
demonstrated empirically in Section~\ref{sec:ablation}, though this
capability has not been fully explored.
	The projection operators $P_k$ encode time-varying pixel
	footprints induced by sensor motion, so that measurements at different times
	sample overlapping but nonidentical regions of the underlying scene. When the ground log--radiance field is discretized on a grid finer than the native image resolution, this overlap introduces additional constraints that support the recovery of sub-pixel spatial detail. From an inverse problems perspective, such super-resolution behavior arises from the structure of the forward operator $AP$ rather than from any modification of the event sensing model. A detailed analysis of identifiability, stability, and achievable resolution gains under different motion patterns and regularization choices is left for subsequent work.

	\paragraph{Nonlinear and non-nadir imaging regimes.}
	The approximate shift-invariance of $P_k$ under near-nadir motion plays a central role in the synthetic-aperture interpretation.
	Extending the analysis to more general flight geometries, wide field-of-view
	sensors, and non-planar terrain remains an important open problem.
	Characterizing how departures from the nadir assumptions affect identifiability and conditioning will be an important step toward extending event-based synthetic-aperture methods to more complex environments.
	
	\bibliographystyle{ieeetr}
	\bibliography{references}
\end{document}